\documentclass[preprint,authoryear,3p, 12pt]{elsarticle}

\usepackage{amssymb}

\usepackage{float}

\journal{Journal of Symbolic Computation}

\newtheorem{theorem}{Theorem}[section]
\newtheorem{prop}[theorem]{Proposition}
\newtheorem{cor}[theorem]{Corollary}
\newtheorem{lemma}[theorem]{Lemma}
\newtheorem{remark}[theorem]{Remark}

\newtheorem{example}[theorem]{Example}

\newenvironment{proof}{\noindent {\em Proof:\ }}{$\square$}

\floatstyle{ruled}
\newfloat{algorithm}{t}{loa}
\floatname{algorithm}{Algorithm}
\newcommand{\Input}[1]
  {\noindent\begin{tabular}{@{}p{1.8cm}@{}p{13.2cm}@{}}
   {\bf Input: }&#1 \end{tabular}}
\newcommand{\Output}[1]
  {\noindent\begin{tabular}{@{}p{1.8cm}@{}p{13.2cm}@{}}
   {\bf Output: }&#1 \end{tabular}}

\def\k{{\rm k}}
\def\K{{\rm K}}

\def\GCD{{\rm gcd}}

\def\Q{{\mathbb{Q}}}
\def\F{{\mathbb{F}}}
\def\A{{\mathcal{A}}}

\def\gr{{Gr\"obner }}

\def\lif{{\bf if \,}}
\def\lthen{{\bf then \,}}
\def\lelse{{\bf else \,}}
\def\lendif{{\bf end if\,}}
\def\lfor{{\bf for \,}}
\def\lfrom{{\bf from \,}}
\def\lendfor{{\bf end for\,}}
\def\ldo{{\bf do \,}}
\def\lto{{\bf to \,}}
\def\lend{{\bf end \,}}
\def\lreturn{{\bf return \,}}
\def\lla{{\leftarrow}}

\def\lbegin{{\bf begin}}

\newcommand{\SPC}{\hspace*{15pt}}
\newcommand{\ignore}[1]{}

\begin{document}

\begin{frontmatter}



\title{An Efficient Algorithm for Factoring Polynomials over
Algebraic Extension Field}


\author{Yao Sun and Dingkang Wang\fnref{label1}}

\fntext[label1]{The authors are supported by NSFC 10971217,
10771206 60821002/F02.}

\address{Key Laboratory of Mathematics Mechanization, Academy of Mathematics and Systems Science, CAS, Beijing 100190,  China}

\ead{sunyao@amss.ac.cn, dwang@mmrc.iss.ac.cn}

\begin{abstract}
A new efficient algorithm is proposed for factoring polynomials
over an algebraic extension field. The extension field is defined
by a polynomial ring modulo a maximal ideal. If the maximal ideal
is given by its \gr basis, no extra \gr basis computation is
needed for factoring a polynomial over this extension field.
Nothing more than linear algebraic technique is used to get a
polynomial over the ground field by a generic linear map. Then
this polynomial is factorized over the ground field. From these
factors, the factorization of the polynomial over the extension
field is obtained.  The new algorithm has been implemented and
computer experiments indicate that the new algorithm is very
efficient, particularly in complicated examples.
\end{abstract}


\end{frontmatter}




\section{ Introduction}


Factorization of polynomials over algebraic extension fields has
been widely investigated and there are polynomial-time algorithms
for factoring multivariate polynomial over algebraic number field
\citep{abbottbd,abbottd,enc1, landau,Lenstra,trager}. However, all
the existing algorithms for factoring polynomials over algebraic
extension field are not so efficient.

Factorization over  algebraic extension fields is needed for
irreducible decomposition of algebraic variety by using
characteristic set  method \citep{wu,wu2}. In
\citep{wangdm,wangl}, Wang and Lin proposed a very good algorithm
for factoring multivariate polynomials over  algebraic fields
obtained from successive extensions of the filed of rational
numbers. This problem has been further investigated by Li and Yuan
in \citep{libh,yuan}. Li's algorithm decomposes ascending chain
into irreducible ones directly and Yuan's algorithm follows
Trager's method \citep{trager}. Their methods   involve the
computation of characteristic set, \gr basis or resultant of
multivariate polynomial system and all these computations are
quite expensive. Rouillier's approach can also deduce an algorithm
for the same aim \citep{Rouillier}. All the above algorithms are
probabilistic, and if the characteristic of the ground field is
$0$, the algorithms terminate in a finite steps with probability 1
\citep{gaochou, wangl}. Besides, A. Steel gave his factorization
method in another way when the characteristic of the field is
positive and he concentrated on how to conquer the inseparability
\citep{Ste05}.

At present, popular methods for factoring polynomials over
extension field are to calculate the primitive element of the
extension field first and factor the polynomials over algebraic
number field afterwards. However, we propose a new factorization
algorithm in a different way. The main purpose of the current
paper is to present a new algorithm to solve the following
factorization problem:

Let $k$ be a perfect computable field and $k[x_1,\cdots,x_n]$ the
polynomial ring in indeterminate $\{x_1,\cdots,x_n\}$ with
coefficients in $k$. Let $I \subset k[x_1, \cdots, x_n]$ be a
maximal ideal such that $\K=k[x_1,\cdots,x_n]/I $ is indeed an
algebraic extension field of $k$. For a polynomial $f \in \K[y]$,
we will derive a new efficient algorithm for factoring $f$ over
the  field $\K$.

The above problem can be converted to univariate polynomial
factorization over the ground field $k$ by using a generic linear
map. If the maximal ideal $I$ is represented by its \gr basis for
any admissible order, no extra \gr basis computation is needed in
the new algorithm.

In \citep{Mon}, Monico proposed a new approach for computing a
primary decomposition of a zero dimensional ideal. This idea also
plays an important role in the new proposed algorithm. However,
Monico's algorithm is not complete, i.e. the components in the
output of Monico's algorithm can not be assured to be primary. The
new algorithm overcomes this flaw when applying Monico's idea to
the above factorization problem, i.e. the irreducible factors can
be verified without extra computations.

This paper is organized as follow. Some necessary preliminaries is
given in section 2. In section 3, we show how the problem of
polynomial factorization over algebraic extension field, which is
proposed in \citep{wangdm, wangl, wu,wu2}, can be transformed to a
univariate factorization problem. A new algorithm for factoring
polynomials over algebraic extension field is presented in section
4. Examples and comparisons appear in section 5 and section 6
respectively. Finally, we conclude this paper in section 7.



\section{Preliminaries}


Let $k$ be a perfect field which admits efficient operations and
factorization of univariate polynomials. Let $R$ be a multivariate
polynomial ring over the field $k$ and $Q$ an ideal of $R$. Let
$\A_k(Q)=R/Q$ denote the quotient ring.

Since we can add elements of $\A_k(Q)$ and multiply elements with
scalars in $k$, $\A_k(Q)$ has the structure of a vector space over
the field $k$. Furthermore, if $Q$ is zero dimensional, then
$\A_k(Q)$ is a finite dimensional vector space.

Given a polynomial $r \in R$, we define a map $m_r$ from $\A_k(Q)$
to itself by multiplication: $$m_r: \A_k(Q) \longrightarrow
\A_k(Q)$$
$$[g]\mapsto [rg],$$ where $[p]$ denotes the class in $\A_k(Q)$ of any
polynomial $p\in R$.

Here are the main properties of the map $m_r$.

\begin{prop} \label{polymap}
Let $r \in R$. Then
\begin{enumerate}

\item \label{p1} $m_r$ is a linear map from $\A_k(Q)$ to
$\A_k(Q)$.

\item \label{p2} $m_r=m_g$ exactly when $r-g \in Q$. In
particular, $m_r$ is the zero map exactly when $r \in Q$.

\item \label{p3} Let $q$ be a univariate polynomial over $k$. Then
$m_{q(r)}=q(m_r)$.

\item \label{p4}  If $p_r$ is the characteristic polynomial of
$m_r$, then $p_r(r) \in Q$.
\end{enumerate}
\end{prop}

\begin{proof}
For the proofs of part (\ref{p1}), (\ref{p2}) and (\ref{p3}),
please see \citep{CLO}. For the part (\ref{p4}), since $p_r$ is
the characteristic polynomial of the linear map $m_r$,
$p_r(m_r)=0$ by Cayley-Hamilton Theorem. According to part
(\ref{p3}), it follows that $m_{p_r(r)}=p_r(m_r)=0$.  Thus,
$p_r(r)$ belongs to the ideal $Q$ by part (\ref{p2}).
\end{proof}

\begin{prop} \label{minpol}
If $Q$ is a maximal ideal of $R$, then the {\em minimal}
polynomial of $m_r$ is irreducible over $k$.
\end{prop}

\begin{proof} Assume $R$ is the polynomial ring $k[x_1,\cdots,x_n]$.
Let $\langle Q, z-r\rangle$ be the ideal generated by $Q$ and
$z-r$ over the polynomial ring $k[x_1,\cdots,x_n, z]$, where $z$
is a new indeterminate. Since $Q$ is a maximal ideal in
$k[x_1,\cdots,x_n]$, it follows that $\langle Q,z-r\rangle$ is
also a maximal ideal in $k[x_1,\cdots,x_n,z]$ and so is the ideal
$\langle Q,z-r\rangle \cap k[z]$.

To study the ideal $\langle Q,z-r\rangle \cap k[z]$, let $g$ be
the monic generator of the principal ideal $\langle Q, z-r\rangle
\cap k[z]$. Substitute the indeterminate $z$ by $r$ in $g$, then
$g(r) \in \langle Q,z-r\rangle \cap k[x_1, \cdots, x_n] = Q$,
which means $g(m_r)=m_{g(r)}=0$ by proposition \ref{polymap}.
Since $\langle Q,z-r\rangle \cap k[z]$ is maximal in $k[z]$, $g$
is irreducible over $k$, and hence $g$ is the minimal polynomial
of $m_r$.
\end{proof}

The following proposition, which is a basic conclusion from
standard linear algebra, illustrates the relationship between {\em
minimal} polynomial and {\em characteristic} polynomial.

\begin{prop}\label{minchar}
The minimal polynomial of $m_r$ and its characteristic polynomial
share the same irreducible factors.
\end{prop}

Thus we have an instant corollary of Proposition \ref{minpol}.

\begin{cor}\label{charpol}
If $Q$ is a maximal ideal of $R$, then the {\em characteristic}
polynomial of $m_r$ is a power of a polynomial which is
irreducible over $k$.
\end{cor}

With the above propositions, next we study more properties about
the characteristic polynomial of $m_r$.

Let now suppose that $Q$ is a zero dimensional radical ideal of
$R$ and $Q$ has a minimal prime decomposition:
$$Q = Q_1 \cap \cdots \cap Q_t,$$ where each $Q_i$ is a prime
ideal of $R$.

We define the linear map $m_{r, i}$ in the same fashion as $m_r$.
Denote $\A_k(Q_i)=R/Q_i$ for $i=1,\cdots,t$, and consider the
linear maps: $$m_{r,i}: \A_k(Q_i) \longrightarrow \A_k(Q_i)$$
$$[g] \mapsto [rg],$$ where $[p]$ denotes the class in $\A_k(Q_i)$ of
any polynomial $p\in R$.

The following proposition proposed by Monico \citep{Mon} describes
the relationship between the characteristic polynomials of $m_r$
and $m_{r,i}$'s.

\begin{prop}\label{cps}
Let $p_r, p_{r,i}$ be the characteristic polynomial of
$m_r,m_{r,i}$ respectively. Then
$$p_{r}=p_{r,1} \cdots p_{r,t}.$$
\end{prop} 

\section{Factorization of Polynomials over Algebraic Extension Field}


In this section, we will discuss the main ideas about the new
factorization method. First of all, we need some new notations.
Throughout this section, let $R=k[x_1,\cdots, x_n]$ and
$R_y=k[x_1,\cdots,x_n,y]$. $I$ is a {\em maximal} ideal in $R$ and
$I_y$ is the ideal generated by $I$ over the polynomial ring
$R_y$. Since $I$ is a maximal ideal, the quotient ring $R/I$ is
indeed a field. For convenience, we denote $\K=R/I$, which is a
finite extension field of $k$. Remark that the quotient ring
$R_y/I_y$ is not a field, as $I_y$ is not a maximal ideal in $R_y$
any more.

The ring $\K[y]$, which is a polynomial ring over $\K$ with the
indeterminate $y$, is a principal ideal domain, so each polynomial
$f$ in $\K[y]$ has a unique factorization over $\K$. What we will
do next is to give an efficient algorithm to calculate the
factorization of $f$ in $\K[y]$.

In order to exploit the properties of $I$, we should connect the
ring $R_y$ and  $\K[y]$. Consider the canonical map:
$$\sigma : R \longrightarrow \K=R/I$$ $$c \longmapsto [c],$$ which sends
a polynomial $c \in R$ to $[c] \in \K$. And $\sigma$ extends
canonically onto $R_y$ by applying $\sigma$ coefficient-wise. By
definition, $\sigma(g)=0$ if and only if $g \in I_y$ for any $g\in
R_y$.

Conversely, given an element $c \in \K$, we say a polynomial $d
\in R$ is a {\em lift} of $c$ if $\sigma(d)=c$. Similarly, we say
$h \in R_y$ is a {\em lift} of $g \in \K[y]$ if $\sigma(h)=g$
holds. Clearly, an element $c \in \K$ (or $g \in \K[y]$) may have
infinite distinct lifts, as the map $\sigma$ is not injective. Pay
attention that, the lifts of $g \in \K[y]$ may have different
degrees in $y$.

Since $\K=k(\alpha_1,
\cdots,\alpha_n)=k[\alpha_1,\cdots,\alpha_n]$, the elements in
$\K$ have polynomial forms in the letters $\alpha_1,
\alpha_2,\cdots,\alpha_n$, i.e. for $ g \in \K[y]$, $g$ has the
following form:
$$g=\sum\limits_{i=0}^d c_i(\alpha) y^i, $$
where $c_i(\alpha) \in \k[\alpha_1,\cdots,\alpha_n]$ for $i=0,
\cdots, d$. Let $h =\sum\limits_{i=0}^d c_i(x) y^i$, where $c_i(x)
\in k[x_1,\cdots,x_n]$ such that $\sigma(c_i(x))=c_i(\alpha)$. It
is easy to check $\sigma(h)=g$, and we call $h$ a {\em natural
lift} of $g$.

Let $F$ be a set of polynomials in $R_y$, the ideal generated by
$F$ over $R_y$ is denoted by $\langle F\rangle_{R_y}$ as usual.

In the rest of this paper, we always make the following
assumptions:
\begin{enumerate}
\item[$\bullet$] $f$ is a squarefree polynomial in $\K[y]$ and
$h\in R_y$ is a lift of $f$.

\item[$\bullet$] $Q=\langle I, h\rangle_{R_y} \subset R_y$ and
$\A_k(Q)=R_y/Q$.

\item[$\bullet$] For $r \in R_y$, the linear map $m_r$ is defined
from $\A_k(Q)$ to $\A_k(Q)$ as in the last section.

\item[$\bullet$] $f=f_1\cdots f_t$ is an irreducible factorization
of $f$ over $\K$ and $h_i$ is a lift of $f_i$.

\item[$\bullet$] $m_{r,i}$ is the linear map defined from
$\A_k(Q_i)$ to $\A_k(Q_i)$, where $\A_k(Q_i)=R_y/Q_i$ and
$Q_i=\langle I, h_i\rangle_{R_y}$ for $i=1,\cdots,t$.

\item[$\bullet$] $p_r$ and $p_{r,i}\in k[\lambda]$ are the
characteristic polynomials of $m_r$ and $m_{r,i}$ respectively.
\end{enumerate}

Now it is time to describe the main ideas of the new algorithm for
factoring $f$ over $\K[y]$. The following lemma builds a relation
between the factorization of a squarefree polynomial and the
minimal decomposition of a radical ideal.

\begin{lemma}\label{prop}
$Q=\langle I, h\rangle_{R_y} \subset R_y$ is a radical ideal and
$$Q=Q_1 \cap \cdots \cap Q_t$$ is a minimal
prime decomposition of $Q$, where $Q_i=\langle I,
h_i\rangle_{R_y}$ for $i=1,\cdots,t$.
\end{lemma}

\begin{proof}
First, we begin by showing the definition of $Q=\langle I,
h\rangle_{R_y}$ is well defined. Suppose $h'$ is another lift of
$f$ in $R_y$. Then it suffices to show the two ideals $Q=\langle
I, h\rangle_{R_y}$ and $Q'=\langle I, h'\rangle_{R_y}$ are
identical. By the definition of lift, we have
$\sigma(h)=f=\sigma(h')$. Since ${\sigma}$ is a homomorphism map,
it follows that ${\sigma}(h-h')=0$, which means $h-h' \in I_y$ and
hence $Q=Q'$. Similarly, $Q_i$'s are also well defined for the
same reasons.

Next, we prove $Q=\langle I, h\rangle_{R_y}$ is a radical ideal of
$R_y$. If $g^m \in Q$ for some positive integer $m$, then $g^m$
has an expression $g^m=t+sh$, where $t \in I_y$ and $s \in R_y$.
Since ${\sigma}$ is a homomorphism map, then
$${\sigma}(g)^m={\sigma}(g^m)={\sigma}(t)+{\sigma}(s){\sigma}(h)={\sigma}(s)f,$$
which means $f  \mid  {\sigma}(g)^m$. Since $f$ is a squarefree
polynomial as assumed, $f \mid {\sigma}(g)^m$ implies $f \mid
{\sigma}(g)$. Let $\sigma(g)= b f $, $b\in \K[y]$ and $a\in R_y$ a
lift of $b$. Since $h$ is a lift of $f$, it follows that
$\sigma(g)=\sigma(a) \sigma(h)$, which means $\sigma(g-ah)=0$ and
hence $g -ah \in I_y$. So $g \in Q$, which shows $Q$ is a radical
ideal.

Similarly, it is easy to show $Q_i$ is a prime ideal by using the
property that $f_i$ is irreducible over $\K$ (hence squarefree),
and the proof is omitted here.

Finally, we finish this proof by showing the $Q_i$'s constitute a
minimal prime decomposition of $Q$.

On one hand, we have $f  \mid  \sigma(g)$ for any $g \in Q$. It
follows that $f_i \mid \sigma(g)$ for $i=1,\cdots,t$. Then $g$
belongs to each $Q_i$ as discussed above and hence lies in the
intersection of these $Q_i$'s.

On the other hand, given $g \in Q_1 \cap \cdots \cap Q_t$, it is
easy to see that $f_i  \mid  {\sigma}(g)$ for all
$i=1,2,\cdots,t$. Since $f_i$'s are irreducible factors of $f$ and
coprime with each other, it follows that $f=f_1f_2\cdots f_t \mid
{\sigma}(g)$, which means there exists $a\in R_y$ such that
$\sigma(g)=\sigma(a)f$ and hence $g-ah \in I_y$. Thus, $g\in Q$.

We have now proved that
$$Q = Q_1 \cap \cdots \cap
Q_t.$$ As $f_i$ and $f_j$ are distinct irreducible factors of $f$
whenever $i\neq j$, then $h_i \notin Q_j$  and $h_j \notin Q_i$,
which indicates the above decomposition is minimal.
\end{proof}

The following theorem is the main theorem of this paper which
provides a new method for factoring polynomials over algebraic
extension fields.

\begin{theorem}[Main Theorem] \label{thm-main}
With the notations defined as earlier. If the characteristic
polynomial $p_r$ of $m_r$ has an irreducible factorization:
$$p_{r}=q_{1}^{m_1} \cdots q_{s}^{m_s},$$ where $q_{i}$ is
irreducible over $k$ and $q_i\not=q_j$ whenever $i\not=j$, then
$\GCD(f, {\sigma}(q_{i}(r)))\not=1$ and
$$f=c \prod_{i=1}^s \GCD (f,\sigma(q_{i}(r))),$$ where $c$ is
constant in $\K$ and $\GCD(g_1,g_2)$ is the monic greatest common
divisor of $g_1$ and $g_2$ for any $g_1,g_2 \in \K[y]$.
Furthermore, if $m_i=1$, then $\GCD(f, {\sigma}(q_{i}(r)))$ is
irreducible over $\K$.
\end{theorem}

\begin{proof}
For convenience, suppose $f$ is monic. In this case, $c=1$.

Since $f$ is squarefree and $f=f_1 \cdots f_t$ is an irreducible
factorization of $f$ as assumed, $Q=\langle I, h\rangle_{R_y}$ is
a radical ideal and $Q_i=\langle I, h_i\rangle_{R_y}$'s are prime
ideals by lemma \ref{prop}. Furthermore, $Q$ has a minimal prime
decomposition $Q = Q_1 \cap \cdots \cap Q_t$. We also have
$p_{r}=p_{r,1} \cdots p_{r,t}$ by proposition \ref{cps}.

Since $p_{r,i} \in k[\lambda]$ is the characteristic polynomial of
$m_{r,i}$, substituting $\lambda$ in $p_{r, i}$ by the expression
of $r\in R_y$, it follows that $p_{r,i}(r)\in Q_i=\langle I,
h_i\rangle_{R_y}$ by proposition \ref{polymap}. That is, there
exist $a\in I_y$ and $b\in R_y$ such that $p_{r,i}(r)=a+bh_i$.
Applying ${\sigma}$ to both sides of equation, we get
${\sigma}(p_{r,i}(r)) =
{\sigma}(a)+{\sigma}(b){\sigma}(h_i)={\sigma}(b)f_i$, which means
$f_i \mid {\sigma}(p_{r,i}(r))$. This shows that $f_i$ is a
nontrivial common divisor of $f$ and $\sigma(p_{r,i}(r))$ for  $1
\le i \le t$.

By corollary \ref{charpol}, each $p_{r,i}$ must be a power of an
irreducible polynomial in $k[\lambda]$. Notice that $p_{r}=p_{r,1}
\cdots p_{r,t}=q_{1}^{m_1} \cdots q_{s}^{m_s}$, which implies that
for each $j$ there exists at least one $p_{r,i}$ such that
$p_{r,i} \mid q_j$. So $\GCD(f, {\sigma}(q_{j}(r)))\not=1$ for $
1\le j \le s$.

We have already shown that $f_i \mid \sigma(p_{r,1}(r)) \cdots
\sigma(p_{r,t}(r))=\sigma(q_{1}(r))^{m_1} \cdots
\sigma(q_{s}(r))^{m_s}$. Since $f_i$ is irreducible over $K$, then
there exists a $j$ where $1\le j \le s$, such that $f_i \mid
\sigma(q_{j}(r))$. As assumed, $f_1,\cdots, f_t$ are distinct
factors of the squarefree polynomial $f$. It follows that
\begin{equation}\label{eq_1} f=f_1\cdots f_t \mid \prod_{i=1}^s \GCD(f,\sigma(q_i(r))).\end{equation}


For each $i$, $\GCD(f,\sigma(q_{i}(r)))$ is  squarefree  since $f$
itself is squarefree.

Since $q_{i}$ and $q_{j}$ are co-prime in $k[\lambda]$ whenever
$i\not=j$, then there exist $a, b \in k[\lambda]$ such that
$aq_{i}+bq_{j}=1$. Substituting $\lambda$ by the expression of
$r$, the equality still holds for $a(r)q_{i}(r)+b(r)q_{j}(r)=1$.
Applying ${\sigma}$ to both sides of equation, we have
${\sigma}(a(r)){\sigma}(q_{i}(r))+{\sigma}(b(r)){\sigma}(q_{j}(r))=1$,
which implies ${\sigma}(q_{i}(r))$ and ${\sigma}(q_{j}(r))$ are
co-prime in $\K[y]$ and hence $\GCD (f,\sigma(q_{i}(r)))$ and
$\GCD (f,\sigma(q_{j}(r)))$ are co-prime as well. Therefore,
$\prod_{i=1}^s \GCD(f,\sigma(q_{i}(r)))$ is squarefree, which
indicates
\begin{equation}\label{eq_2}
\prod_{i=1}^s \GCD(f,\sigma(q_{i}(r))) \mid f,
\end{equation} since $\GCD(f,\sigma(q_{i}(r)))\mid f$ for $1\le i\le
s$.

From (\ref{eq_1}) and (\ref{eq_2}), we have  $f=\prod_{i=1}^s \GCD
(f,\sigma(q_{i}(r)))$. The first part of theorem is proved.


Particularly, if $m_k=1$ for some $k$, the equation $q_{1}^{m_1}
\cdots q_{s}^{m_s}=p_{r,1} \cdots p_{r,t}$ shows $q_{k}$ divides
only one $p_{r, i}$. Then we have $p_{r, i}=q_{k}$ and $p_{r, i}$
is co-prime with other $p_{r, j}$ whenever $i\not=j$. With a
similar discussion, it is easy to show $\sigma(p_{r, i}(r))$ and
$\sigma(p_{r, j}(r))$ are co-prime in $\K[y]$ whenever $i\not=j$.
Clearly, $\GCD (f,\sigma(p_{r, i}(r)))=\GCD (f,\sigma(q_{k}(r)))$
is a factor of $f$ and we also know $f_i \mid \GCD (f,\sigma(p_{r,
i}(r)))$ as discussed earlier. Therefore, if there exists $f_j$
such that $f_i\not=f_j$ and $f_j\mid \GCD (f,\sigma(p_{r,
i}(r)))$, then $\sigma(p_{r, i})$ and $\sigma(p_{r, j})$ will have
a nontrivial common divisor $f_j$. This contradiction implies
$\GCD (f,\sigma(q_{k}(r)))=f_i$ and hence irreducible over $\K$.
\end{proof}

Then we have two immediate corollaries of the main theorem.

\begin{cor} \label{cor-standard}
If the characteristic polynomial $p_r$ of $m_r$ is squarefree,
suppose $p_r= q_{1} \cdots q_{s}$ is an irreducible factorization
of $p_r$ over $k$, then
$$f=c \prod_{i=1}^s  \GCD (f,\sigma(q_{i}(r)))$$ is an irreducible
factorization of $f$ over $\K$, where $c$ is constant in $\K$.
\end{cor}

\begin{cor}
If the characteristic polynomial of $m_r$ is irreducible over $k$,
then $f$ is irreducible over $\K$.
\end{cor}

Corollary \ref{cor-standard} indicates that if we are lucky enough
to get a squarefree characteristic polynomial $p_r$, then we can
obtain the complete factorization of $f$ directly; otherwise, by
the main theorem \ref{thm-main}, we will get some factors of $f$,
which can be factored in a further step.

The most important contribution of the main theorem is that we are
able to check which factor of $f$ is irreducible by simply
investigating whether $m_i$ is $1$, which ensures the method
provided in this paper is a complete method for factoring
polynomials in $\K[y]$.


\section{Algorithm for Factorization}


In this section, we will present the algorithm for factorization
over algebraic extension field based on the main theorem
\ref{thm-main}. Before doing that, we discuss some algorithmic
details first.

Given a polynomial $f\in \K[y]$, it is usually not squarefree. So
in order to apply the main theorem, we can factor the squarefree
part of $f$ first and deduce a factorization of $f$ afterwards,
which is not very difficult no matter the field $K$ is
characteristic 0 or not. In the new algorithm, the $\GCD$
computation over algebraic extension field is necessary, and many
algorithms have been proposed for this purpose \citep{Hoeij,lange,
maza}.

In case the {\em characteristic} polynomial of $m_r$ is difficult
to compute, we can calculate the {\em minimal} polynomial of $m_r$
instead with the following observation.

\begin{prop}\label{charequalmin}
If the {\em characteristic} polynomial of $m_r$ is squarefree,
then the {\em minimal} polynomial and the {\em characteristic}
polynomial of $m_r$ are identical.
\end{prop}

\begin{proof}
It is an easy corollary of proposition \ref{minchar}.
\end{proof}

Conversely, if the minimal polynomial has lower degree than its
characteristic polynomial, then the characteristic polynomial is
not squarefree. Many methods can be exploited for computing the
minimal polynomial, such as the famous {\em FGLM} method
\citep{FGLM}.

Now, it is time to present the algorithm for factorization over
algebraic extension field.

\begin{algorithm}[!ht] \caption{\bf --- Factorization}
  \smallskip
  \Input{$f$, a squarefree monic polynomial in $\K[y]$.}\\
  \Output{the factorization of $f$ in $\K[y]$.}
  \medskip

  \noindent
  \lbegin\\
  \SPC $r \lla$ a random polynomial in $k[x_1,\cdots,x_n,
  y]$\\
  \SPC $p_r \lla$ the characteristic polynomial of the
  linear map $m_r$\\
  \SPC factor $p_r$ over $k$ and obtain $p_r=q_1^{m_1}\cdots
  q_s^{m_s}$\\
  \SPC \lfor $i$ \lfrom $1$ \lto $s$ \ldo \\
  \SPC\SPC $f_i \lla \GCD(f,\sigma(q_i(r)))$\\
  \SPC\SPC \lif $m_i= 1$   $\# f_i$ is irreducible\\
  \SPC\SPC\SPC \lthen $g_i\lla f_i$\\
  \SPC\SPC\SPC \lelse $g_i \lla  Factorization(f_i)$\\
  \SPC\SPC \lendif\\
  \SPC \lendfor\\
  \SPC \lreturn $g_1g_2\cdots g_s$\\
  \lend
\end{algorithm}

\begin{remark}
The computation of characteristic polynomial $p_r$ of $m_r$ is an
important step of the above algorithm. According to the method
provided in \citep{CLO}, $p_r$ is easy to compute if the \gr basis
of $Q=\langle I, h\rangle_{R_y}$ is known, where $h$ is a nature
lift of $f$. Fortunately, if $f$ is monic in $\K[y]$, the \gr
basis of $Q$ can be constructed directly, since $\{G, h\}$ is a
\gr basis of $\langle I, h\rangle_{R_y}$ with the elimination
monomial order $y \succ x$, where $G$ is a \gr basis of $I$.
\end{remark}

The correctness of the above algorithm is ensured by the main
theorem \ref{thm-main}. So it remains to discuss the termination.

First, we will show the characteristic polynomial $p_r$ of $m_r$
is squarefree with a fairly high probability for a random chosen
$r\in R_y$. Clearly, if $p_r$ is squarefree, then the algorithm
terminates immediately by corollary \ref{cor-standard}.

\begin{prop}\label{prop-probability}
If the characteristic of $k$ is $0$, then the probability that the
characteristic polynomial $p_r$ of $m_r$ is squarefree for  a
random $r \in R_y$  is $1$.
\end{prop}

\begin{proof} The technique of the proof draws lessons from
\citep{Mon}.

Since $Q=\langle I, h\rangle_{R_y}$ is a zero dimensional radical
ideal, the quotient ring $\A_k(Q)=R_y/Q$ has finite dimension as a
vector space. Let $d=\dim_k(\A_k(Q))$. According to the basic
algebraic geometry, we know the variety $V(Q)$ has $d$ distinct
points, say $z_1,\cdots,z_d$, in an extension field of $k$.

Notice that $p_r\in k[\lambda]$ is squarefree if and only if
$r(z_i) \not= r(z_j)$ whenever $i \not= j $, which is a direct
consequence of theorem 4.5 in \citep{CLO}. Therefore, consider the
set: $$C=\{ r \mid p_r \mbox{ is not squarefree}\}=\{ r \mid
\exists z_i,z_j\in V(Q) \mbox{ with } z_i\not=z_j \mbox{ such that
} r(z_i)=r(z_j)\}.$$ Since $V(Q)$ has finite points, then it only
suffices to show the set
$$C_{ij}=\{ r \mid r(z_i)=r(z_j) \mbox{ and } z_i\not= z_j\}$$ is an algebraic set.

Let $\{ e_1,\cdots,e_d \}$ be the standard monomial basis of
$\A_k(Q)$. Thus, $[r]=a_1e_1+\cdots+a_de_d$, where $a_i\in k$ for
$i=1 \cdots d$. So $C_{ij}$ also has an isomorphic form:
$${\tilde{C}_{ij}}=\{(a_1,\cdots,a_d)\in k^d \mid a_1e_1(z_i)+\cdots+a_de_d(z_i)=a_1e_1(z_j)+\cdots+a_de_d(z_j) \mbox{ and } z_i\not= z_j\}.$$
According to the section 2.4 of \citep{CLO}, $z_i$ is uniquely
determined by the vector $(e_1(z_i),\cdots,e_d(z_i))$. Therefore,
$z_i\not= z_j$ implies
$$(e_1(z_i),\cdots,e_d(z_i))\not=(e_1(z_j),\cdots,e_d(z_j))$$ and
hence $(e_1(z_i)-e_1(z_j),\cdots,$ $e_d(z_i)-e_d(z_j))$ is a
nonzero vector.

Thus $ \tilde{C}_{ij} $ is a proper algebraic set in $k^d$.
Consequently, $C$ is isomorphic to a proper algebraic set of
$k^d$. Since the characteristic of $k$ is $0$, the probability
that a random $r \in R_y$  belongs to the set $C$ is $0$, which
completes the proof.
\end{proof}

In order to simplify the computation, we usually prefer $r$ in a
linear form. The following corollary shows the characteristic
polynomial $p_r$ of $m_r$ is also squarefree with a high
probability for a randomly chosen {\em linear} $r$.

\begin{cor}
If the characteristic of $k$ is $0$, then the probability that the
characteristic polynomial $p_r$ of $m_r$ is squarefree for  a
random {\em linear} $r \in R_y$ is also $1$.
\end{cor}

\begin{proof}
The proof is in the same fashion as proposition
\ref{prop-probability}. The mere difference is that $r$ has a {\em
linear} expression $r=by+a_1x_1+\cdots+a_nx_n$. Then the set
$C_{ij}=\{ r \mid r(z_i)=r(z_j) \mbox{ and } z_i \not= z_j\}$ is
isomorphic to a proper algebraic set of $k^{n+1}$, which completes
the proof.
\end{proof}

There are some tricks for choosing a linear $r$ so as to speed up
the algorithm. For example, the variable $y$ needs to appear in
the expression of $r$ and we usually set the coefficient of $y$ as
1; also, if the variable $x_i$ happens to be a leading power
product of some polynomial in the \gr basis of $I$, then this
variable $x_i$ is not needed in $r$, as it can be reduced
afterwards.

Although the probability that the characteristic polynomial $p_r$
of $m_r$ is squarefree for a random (linear) $r \in R_y$ is $1$,
it is not sufficient to show the algorithm terminates all the
time. However, the following proposition indicates that if we
select $r$ in a special fashion, the algorithm terminates in
finite steps.

\begin{prop}
If the characteristic of $k$ is $0$, then we can find an $r\in
R_y$ such that $p_r$ is squarefree in finite steps.
\end{prop}

\begin{proof}
In fact, according to the proof of proposition
\ref{prop-probability}, the set $C$ is the union of all $C_{ij}$
for $i\not=j$, where $C_{ij}$ is isomorphic to the set
$\{(a_1,\cdots,a_d)\in k^d \mid
a_1(e_1(z_i)-e_1(z_j))+\cdots+a_d(e_d(z_i)-e_d(z_j))=0 \mbox{ and
} z_i\not= z_j\}$. Thus, $C$ is isomorphic to the solution set of
a polynomial equation $F(a_1,\cdots,a_d)=0$, while the total
degree of $F$ is at most $d(d-1)/2$. Let
$d_i=\deg_{a_i}F(a_1,\cdots,a_d)$ for $i=1,\cdots,d$ and
$D=\{(a_1, \cdots, a_d)\mid a_i=0, \cdots, d_i \mbox{ for } 1\le
i\le d\}$. Since $F\not=0$, $F$ cannot vanish on all the points of
$D$. So there exist $(a'_1, \cdots, a'_d)\in D$ such that $F(a'_1,
\cdots, a'_d)\not=0$. Then $r=a'_1e_1+\cdots+a'_de_d$ is the $r$
such that $p_r$ is squarefree. As the cardinality of $D$ is
finite, this $r$ can be constructed within finite steps.
\end{proof}

Therefore, in each recursive call of {\em Factorization$(f_i)$},
if we choose a different $r$ in the above fashion, the algorithm
must terminate in finite steps.

At last, let say something about the complexity of the new
algorithm. Given a \gr basis $G$ of $I$ and the set $\{G, h\}$ is
a \gr basis of $Q=\langle I, h\rangle_{R_y}$ as discussed earlier,
so computing a basis for $\A_k(Q)$ has complexity $O(n)$.
Computing the matrix of $m_r$ requires $O(n^3)$ field operations
in the worst case. Computing the characteristic polynomial $p_r$
requires $O(n^3)$ field operations. Factoring the univariate
polynomial $p_r$ has been studied by many researchers, and more
details can be found in \citep{coh,LLL}. As a result, by using
this new algorithm, the problem of factoring polynomials over
algebraic extension field can be transformed to the factorization
of univariate polynomials over the ground field in polynomial
time.


\section{A Complete Example}


In this section, we illustrate the new algorithm through a
complete example.

\begin{example}
Given a maximal ideal $I=\langle x_1^2+1, x_2^2+x_1\rangle \subset
\Q [x_1,x_2]$, where $\Q$ is the rational field.  Then the
extension field is $\K=\Q[x_1,x_2]/I$. Notice that
$\{x_1^2+1,x_2^2+x_1\}$ is already a \gr basis of $I$ for the
lexicographic order with $x_2\succ x_1$.

We are going to factor the polynomial
$$f=y^3+(\alpha_1\alpha_2-2\alpha_1-\alpha_2)y^2+(\alpha_1\alpha_2+2\alpha_2-2)y+\alpha_1-\alpha_1\alpha_2\in
\K[y],$$ where $\alpha_i=[x_i] \in \K$.
\end{example}

Since $f$ is squarefree and monic in $\K[y]$,
$$h=y^3+(x_1x_2-2x_1-x_2)y^2+(x_1x_2+2x_2-2)y+x_1-x_1x_2\in R_y$$
is a natural lift of $f$. Thus, $\{x_1^2+1, x_2^2+x_1, h\}$ is a
\gr basis of the ideal $Q=\langle I, h\rangle_{\Q[x_1, x_2, y]}$
for the lexicographic order with $y \succ x_2 \succ x_1$.

According to the new algorithm, we need to choose a random
polynomial $r\in R_y=\Q[x_1,x_2,y]$ first. Here $r=x_1+2x_2+y$ is
selected. Let $\A_k(Q)=\Q[x_1,x_2,y]/Q$, which is obviously a
vector space over $\Q$ with a monomial basis
$$B=[1, x_2, x_1, x_1x_2, y, x_2y, x_1y, x_1x_2y, y^2, x_2y^2,
x_1y^2, x_1x_2y^2]^T.$$

Next, compute the matrix $M$ of the linear map $m_r$ w.r.t. $B$.
Then
$$m_r(B)=MB,$$ where $M$ is  a $12 \times 12$ matrix
$$M=\left( \begin{array}{cccccccccccc}
               0& 2& 1& 0& 1& 0& 0& 0& 0& 0& 0& 0\\
               0& 0& -2& 1& 0& 1& 0& 0& 0& 0& 0& 0\\
               -1& 0& 0& 2& 0& 0& 1& 0& 0& 0& 0& 0\\
               2& -1& 0& 0& 0& 0& 0& 1& 0& 0& 0& 0\\
               0& 0& 0& 0& 0& 2& 1& 0& 1& 0& 0& 0\\
               0& 0& 0& 0& 0& 0& -2& 1& 0& 1& 0& 0\\
               0& 0& 0& 0& -1& 0& 0& 2& 0& 0& 1& 0\\
               0& 0& 0& 0& 2& -1& 0& 0& 0& 0& 0& 1\\
               0& 0& -1& 1& 2& -2& 0& -1& 0& 3& 3& -1\\
               1& 0& 0& -1& -1& 2& 2& 0& -1& 0& -3& 3\\
               1& -1& 0& 0& 0& 1& 2& -2& -3& 1& 0& 3\\
               0& 1& 1& 0& -2& 0& -1& 2& 3& -3& -1& 0\\
      \end{array}\right) $$
The characteristic polynomial of this matrix is
$$p_r=\lambda^{12}+26\lambda^{10}-116\lambda^9+371\lambda^8-2064\lambda^7+6802\lambda^6-17916\lambda^5
+49922\lambda^4$$
$$-109088\lambda^3+155984\lambda^2-134592\lambda+55872 {\hspace*{100pt}}$$
$$=(\lambda^4+10\lambda^2-12\lambda+18)(\lambda^4+8\lambda^2-72\lambda+97)(\lambda^4+8\lambda^2-32\lambda+32).$$

The next step is to substitute $\lambda$ by the expression of $r$
in each factor of $p_r$. For instance,
$q_{1}=\lambda^4+10\lambda^2-12\lambda+18$ becomes
$$q_{1}(r)=(x_1+2x_2+y)^4+10(x_1+2x_2+y)^2-12(x_1+2x_2+y)+18.$$
And
$$\sigma(q_{1}(r))=(\alpha_1+2\alpha_2+y)^4+10(\alpha_1+2\alpha_2+y)^2-12(\alpha_1+2\alpha_2+y)+18\in \K[y].$$

In the following, we compute the $\GCD$ of $f$ and
$\sigma(q_{1}(r))$. Finally obtain
$$\GCD(f, \sigma(p_{r,1}(r)))=y+\alpha_1\alpha_2.$$
Since $m_1=1$, $y+\alpha_1\alpha_2$ is an irreducible factor of
$f$ by theorem \ref{thm-main}. Similarly, since $m_2=m_3=1$, the
other irreducible factors of $f$ can be obtain from
$q_{2}=\lambda^4+8\lambda^2-72\lambda+97$ and
$q_{3}=\lambda^4+8\lambda^2-32\lambda+32$:
$$\GCD(f, \sigma(q_{2}(r)))=y-\alpha_1-\alpha_2, \mbox{ and }
\GCD(f, \sigma(q_{3}(r)))=y-\alpha_1.$$ As a result, we get a
complete factorization of $f\in \K[y]$:
$$f=(y+\alpha_1\alpha_2)(y-\alpha_1-\alpha_2)(y-\alpha_1).$$

In the above procedure, $p_r$ is squarefree, so we obtain a
complete factorization of $f$ directly. However, what if $p_r$ is
not squarefree?

For example, if $r=-\frac{3}{2}x_1-\frac{1}{2}x_2+y$ is selected
at the beginning, then we repeat the above steps.

The monomial basis $B$ does not change, but the matrix varies and
the characteristic polynomial becomes
$$p_r=\lambda^{12}+\frac{7}{2}\lambda^{10}-\frac{7}{2}\lambda^9
+\frac{113}{8}\lambda^8-\frac{3}{2}\lambda^7+\frac{33}{4}\lambda^6
+\frac{41}{4}\lambda^5+\frac{273}{64}\lambda^4+\frac{67}{16}\lambda^3
+\frac{467}{128}\lambda^2+\frac{169}{128}\lambda+\frac{89}{512}$$
$$=(\lambda^4+\frac{5}{2}\lambda^2-\frac{9}{2}\lambda+\frac{89}{8})
(\lambda^4+\frac{1}{2}\lambda^2+\frac{1}{2}\lambda+\frac{1}{8})^2.$$

Let
$q_1=\lambda^4+\frac{5}{2}\lambda^2-\frac{9}{2}\lambda+\frac{89}{8}$
and
$q_2=\lambda^4+\frac{1}{2}\lambda^2+\frac{1}{2}\lambda+\frac{1}{8}$.
Since $m_1=1$, we can get an irreducible factor of $f$ by theorem
\ref{thm-main}:
$$\GCD(f, \sigma(q_{1}(r)))=y+\alpha_1\alpha_2.$$ While the other
factor $q_2$ only leads to a reducible factor of $f$:
$$\GCD(f, \sigma(p_{r,2}(r)))=y^2-(2\alpha_1+\alpha_2)y+\alpha_1\alpha_2-1,$$
which needs to be factored further.

Let $f'=y^2-(2\alpha_1+\alpha_2)y+\alpha_1\alpha_2-1$ and
$h'=y^2-(2x_1+x_2)y+x_1x_2-1\in R_y$ is a natural lift of $f'$.
Next $r'=-2x_1-2x_2+y$ is chosen. And the monomial basis of
$\Q[x_1, x_2, y]/\langle I, h'\rangle_{\Q[x_1, x_2, y]}$ is
$$B'=[1, x_2, x_1, x_1x_2, y, yx_2, yx_1, yx_1x_2]^T.$$ Notice the
length of $B'$ is $8$, which is smaller than the previous one.
Thus an $8 \times 8$ matrix is constructed and the characteristic
polynomial is
$$p_{r'}=\lambda^8+4\lambda^6+20\lambda^5+23\lambda^4+40\lambda^3+102\lambda^2+100\lambda+34$$
$$=(\lambda^4+2\lambda^2+16\lambda+17)(\lambda^4+2\lambda^2+4\lambda+2)=q'_{1}q'_{2}.$$

Since $m'_1=m'_2=1$, then we obtain two irreducible factors of
$f'$:
$$\GCD(f', \sigma(q'_{1}(r')))=y-\alpha_1, \mbox{ and } \GCD(f', \sigma(q'_{2}(r')))=y-\alpha_1-\alpha_2.$$

Combined with the factor we got earlier, $f$ has a complete
factorization in $\K[y]$:
$$f=(y+\alpha_1\alpha_2)(y-\alpha_1)(y-\alpha_1-\alpha_2).$$

The new algorithm can also perform very well when the ground field
$k$ is a finite field. However, if we consider the factorization
when the ground field is a finite field, according the proof of
proposition \ref{prop-probability}, we will have a lower
probability to find an $r$ such that $p_r$ is squarefree,
especially when the cardinality of $k$ is small.

\ignore{
\begin{example}
Let $k=\F_{31}$, $I=\langle x_1^2+5, x_2^3+x_1\rangle \subset
\F_{31}[x_1,x_2]$ and $\K=\F_{31}[x_1,x_2]/I$. The polynomial to
be factored is
$$f=y^3+(\alpha_1+2\alpha_2-\alpha_1\alpha_2)y^2+(5\alpha_2+2\alpha_1\alpha_2-2\alpha_2^2\alpha_1)y+10\alpha_2^2\in \K[y],$$
where $\alpha_i=[x_i] \in \K$ for $i=1,2$.
\end{example}

Let $h$ be the natural lift of $f$, $Q=\langle I,
h\rangle_{\F_{31}[x_1,x_2, y]}$ and $\A_k(Q)= R_y/Q$. Next
$r=-x_1+x_2+y$ is chosen, while the characteristic polynomial of
$m_r$ with respect to the standard monomial basis of $\A_k(Q)$ is
$$p_r=\lambda^{18}+28\lambda^{16}+20\lambda^{15}+9\lambda^{14}+22\lambda^{13}
+19\lambda^{12}+24\lambda^{11}+5\lambda^{10}+11\lambda^9+22\lambda^8+6\lambda^7
+11\lambda^6+15\lambda^5$$
$$+23\lambda^3+24\lambda^2+20\lambda+28 {\hspace*{260pt}}$$
$$=(\lambda^6+15\lambda^4+14\lambda^2-6)
(\lambda^6+15\lambda^4-11\lambda^3-4\lambda^2+10\lambda+14)
(\lambda^6-2\lambda^4-11\lambda^2+10).$$

Since $p_r$ is squarefree, then we can immediately obtain a
complete factorization of $f$ by corollary \ref{cor-standard}:
$$f=(y+2\alpha_2)(y-\alpha_1\alpha_2)(y+\alpha_1).$$}


\section{Timings}


We have implemented the new algorithm both for the case $k=\Q$ and
for finite fields in {\em Magma}. Since Wang's algorithm can only
work for fields of characteristic 0. In order to be fair, the
examples are randomly generated over the ground field $k=\Q$.

We tested the examples in appendix both for {\em cfactor} which is
an implementation of Wang's algorithm and for {\em efactor} which
is an implementation of the new algorithm. The timings in the
following table are obtained from a computer (Windows XP, CPU
Core2 Duo 2.66GHz, Memory 2GB).

We should mention that {\em cfactor} is implemented in {\em Maple
7}, since {\em cfactor} only can work correctly for {\em Maple 7},
while {\em efactor} is implemented in {\em Magma}. For the input
of the new algorithm, the maximal ideal can be expressed by its
\gr basis for any admissible order, generally for a total degree
order. And for the input of Wang's algorithm, the maximal ideal
has to be its irreducible ascending set, which is equivalent to a
lexicographic \gr basis. Notice that a \gr basis with
lexicographic order usually has larger coefficients than that with
a total degree order.

{\small
\begin{table}[!ht]\centering
\begin{tabular}{|c|c|c|c|c|c|c|c|}\hline
   & $R_y $ & $\dim_k R_y/\langle I^{(i)}, h^{(i)}\rangle_{R_y}$ & \em{cfactor(sec.)} & \em{efactor(sec.)}\\ \hline

$f^{(1)}$ & $\Q[x_1,x_2,y]$ &16 & 0.032 & 0.000\\ \hline

$f^{(2)}$ & $\Q[x_1,x_2,x_3,y]$&28 & 0.110 & 0.031\\ \hline

$f^{(3)}$ & $\Q[x_1,x_2,x_3,x_4,y]$&48 & 12.171 & 0.734\\ \hline

$f^{(4)}$  & $\Q[x_1,x_2,x_3,x_4,y]$ &32& 9.109 & 0.328\\ \hline

$f^{(5)}$  & $\Q[x_1,x_2,x_3,x_4,y]$ & 64& 245.531 & 4.313\\
\hline

$f^{(6)}$  & $\Q[x_1,x_2,x_3,x_4,x_5,y]$ & 32& 44.359 & 1.297\\
\hline

$f^{(7)}$  & $\Q[x_1,x_2,x_3,x_4,x_5,y]$ & 48&91.500 & 9.719\\
\hline

$f^{(8)}$  & $\Q[x_1,x_2,x_3,x_4,x_5,y]$ & 48& 377.327 & 11.469\\
\hline

$f^{(9)}$  & $\Q[x_1,x_2,x_3,x_4,x_5,y]$ & 80 &2011.375 & 63.578\\
\hline

$f^{(10)}$  & $\Q[x_1,x_2,x_3,x_4,x_5,x_6,y]$ & 64& $>2h$ & 96.344\\
\hline
\end{tabular}
\caption{Compared with Wang's Algorithm.} \label{tab-1}
\end{table}}

In the third column of the above table, $h^{(i)}$ is a lift of
$f^{(i)}$. From this table, we can see that the new algorithm is
much more efficient than Wang's, especially for complicated
examples.

By analyzing Wang's algorithm and the new algorithm, we think
there are three main reasons that make the new algorithm more
efficient than Wang's.  First, in Wang's algorithm, the variable
$y$ in $f$, which is to be factored, needs to be replaced by a
linear combination of a new variable $y'$ and the $x_i$'s. This
leads to the expansions of the coefficients as well as the terms
of $f$ when the degree of $f$ in $y$ is big. Second, the modulo
map by a \gr basis, which sends a polynomial into its remainder,
is a ring homomorphism, which speeds up the new algorithm. But in
Wang's algorithm, the psudo-remainder map does not hold this
property. Last and the most important, the complexity of computing
the characteristic polynomial of $m_r$ is polynomial time for any
given $r$. However, the complexity of computing the characteristic
set in Wang's algorithm is exponential. Besides, any new technique
for calculating the characteristic polynomial will speed up the
new algorithm as well.


\section{Conclusions and Future Works}


In this paper, we present a new method for factoring polynomials
over an algebraic extension field and this algorithm performs
pretty good for characteristic 0 systems as well as finite field
systems. Compared with Monico's primary decomposition method, the
new algorithm is complete and the irreducible factors can be
verified without extra computations. The new algorithm surely
terminates within finite steps if the linear map in each recursive
call of the algorithm is selected in a special fashion. And in
most cases, the proposed algorithm terminates in few loops, as the
characteristic polynomial of a generic linear map is squarefree
with probability $1$. Moreover, the total complexity of this new
algorithm can be controlled in a reasonable degree.

However, when the characteristic of ground field is 0, the
expansion of coefficients is unavoidable. The situation is better
in finite field. Therefore, a natural idea emerges. That is we can
factor the polynomials in finite field first, and lift the
factorization to characteristic 0 afterwards. We also notice that
Gao gives an efficient algorithm for computing the primary
decomposition over finite fields \citep{gaoshuhong}, which may
help to improve the new algorithm in finite field and hence
benefits for our future work.


\section{Acknowledgements}


We would like to thank Professor D. Lazard and Professor V.P.
Gerdt for their valuable  suggestions during their visits in KLMM.

\appendix

\section{Examples in Timings}

\begin{itemize}
   \item[1.]
   $f^{(1)}=(y+\alpha_1)(y-2\alpha_2)(y^2+\alpha_1+\alpha_2)$,\\
   $I^{(1)}=(x_1+x_2^2,1+x_1^2-x_2x_1)\subset \Q[x_1,x_2]$.

   \item[2.]
   $f^{(2)}=(y+\alpha_1\alpha_3+\alpha_2+\alpha_1)(y-2\alpha_2^2+\alpha_3^2+1)(y^2+\alpha_1\alpha_2+\alpha_3)$,\\
   $I^{(2)}=(x_1^2-x_2x_1+x_3x_1-x_1-x_2, x_1^2-x_3x_1+x_1-x_2^2-x_3x_2-x_2+x_3^2, -1+x_1^2+x_3x_1+x_1-x_2^2-x_2+x_3^2-x_3)\subset
   \Q[x_1,x_2,x_3]$.

   \item[3.]
   $f^{(3)}=(y+\alpha_1)(y-2\alpha_4)(y+\alpha_2+\alpha_3)$,\\
   $I^{(3)}=(x_1^2+x_3x_1-x_1x_4+x_2^2-x_2+x_3x_4-x_4^2-x_4, x_1^2+x_2x_1-x_1x_4+x_2^2+x_3x_2+x_2+x_4^2-x_4, 1+x_2x_1-x_3x_1+x_1x_4+x_1+x_2^2-x_3x_2+x_2x_4-x_2+x_3^2+x_3x_4-x_4^2, x_1^2+x_2x_1+x_3x_1+x_1x_4-x_2^2+x_3x_2-x_2+x_3^2-x_3x_4-x_4^2+x_4)\subset
   \Q[x_1,x_2,x_3,x_4]$.

   \item[4.]
   $f^{(4)}=(y-2\alpha_4^2+\alpha_3\alpha_1+\alpha_2+1)(y+\alpha_2^2+\alpha_3\alpha_4+\alpha_1\alpha_3+2)$,\\
   $I^{(4)}=(-1-x_1^2+x_3x_1+x_2^2-x_3x_2+x_3^2-x_3x_4+x_4^2+x_4, 1+x_2x_1+x_3x_1+x_1-x_3x_2-x_2x_4+x_2-x_3^2-x_3x_4-x_3, 1+x_3x_1+x_1x_4+x_1+x_2^2+x_2x_4-x_2-x_3-x_4^2, x_1^2+x_2x_1+x_3x_1+x_1x_4-x_1-x_2x_4-x_3^2+x_3x_4-x_3+x_4^2+x_4)\subset
   \Q[x_1,x_2,x_3,x_4]$.

   \item[5.]
   $f^{(5)}=(y^2+(\alpha_1+\alpha_4\alpha_2)y+\alpha_3\alpha_4+\alpha_2)(y^2+(\alpha_1\alpha_3-\alpha_4)y+\alpha_3+\alpha_2\alpha_4\alpha_1)$,\\
   $I^{(5)}=(-1+2x_1^2-x_2x_1+2x_3x_1-x_1x_4+x_2^2+x_2x_3+2x_2x_4-2x_3^2+2x_3x_4-x_3-x_4, x_1^2-2x_1x_4-x_1+2x_2^2+x_2x_3+2x_2x_4-x_2-x_3x_4+x_3-2x_4^2+2x_4, 2-2x_1^2+2x_2x_1+x_3x_1+2x_1+2x_2x_3+x_2x_4-x_3^2-2x_3-2x_4^2, 2x_1^2-x_2x_1-x_3x_1-x_1x_4-x_2^2+x_2x_3-x_2x_4-2x_2+2x_3^2-2x_3x_4+x_3+x_4^2+2x_4)\subset
   \Q[x_1,x_2,x_3,x_4]$.

   \item[6.]
   $f^{(6)}=(y+\alpha_1+\alpha_3\alpha_4+\alpha_2\alpha_5)(y+\alpha_2\alpha_5+2-\alpha_3+\alpha_4)$,\\
   $I^{(6)}=(2-x_1^2+x_1x_2-2x_3x_1+2x_4x_1-2x_1+2x_2^2+x_2x_3+2x_2x_5+2x_3^2+x_3x_4-x_3-2x_4x_5-2x_4+2x_5, 1-x_1^2-x_1x_2-x_3x_1-x_1x_5+x_1-x_2^2-2x_2x_3+2x_2x_4-2x_2x_5+x_2+x_3^2+2x_3x_4-x_3-2x_4^2-x_4x_5+x_4+2x_5^2, 1-2x_1^2-2x_1x_2-2x_3x_1-x_4x_1-x_1x_5+x_2^2+2x_2x_3-2x_2x_4+x_2x_5-2x_2-x_3^2-2x_3x_4+2x_3x_5-2x_3-2x_4^2+x_4x_5+2x_5^2+x_5, x_1^2-x_1x_2-x_3x_1-x_4x_1+2x_1+2x_2^2+x_2x_3+2x_2+2x_3x_5+x_4^2-x_4+2x_5^2-x_5, x_2-2x_4+x_5-1)\subset
   \Q[x_1,x_2,x_3,x_4,x_5]$.

   \item[7.]
   $f^{(7)}=(y+\alpha_1+\alpha_3\alpha_4)(y-\alpha_2\alpha_5)(y-\alpha_3+\alpha_4)$,\\
   $I^{(7)}=(1-2x_3x_1+x_4x_1+2x_1x_5+x_1-2x_2^2-x_2x_3+2x_2x_4+2x_2x_5-2x_2-x_3^2-2x_3x_4+2x_3x_5+x_3+x_4^2+2x_4x_5+x_4-x_5^2-2x_5, 1-2x_1^2+2x_1x_2-2x_3x_1-2x_4x_1+2x_1x_5-2x_1+x_2^2+2x_2x_3+x_2x_4+2x_2x_5-2x_2+2x_3^2-x_3x_4+2x_3x_5+2x_3+2x_4^2-x_4x_5-2x_5^2+x_5, -x_1^2-x_1x_2-x_3x_1+x_4x_1-2x_1x_5+2x_1-x_2^2+2x_2x_3-x_2x_4+2x_2x_5+2x_2-2x_3^2+2x_3x_4-x_3x_5-x_3-x_4^2+2x_4x_5+2x_4-2x_5^2+2x_5, -1-2x_1^2+2x_1x_2-x_3x_1-x_4x_1+x_1x_5+x_1+2x_2x_3+x_2x_4+x_2x_5+2x_2+x_3x_4+x_3-2x_4^2-x_4+x_5^2-x_5, x_2-x_3+x_4-x_5+1)\subset
   \Q[x_1,x_2,x_3,x_4,x_5]$.

   \item[8.]
   $f^{(8)}=(y^2+(\alpha_1-\alpha_2\alpha_4)y+\alpha_2\alpha_5+\alpha_3+\alpha_5)(y+\alpha_3\alpha_5+\alpha_2\alpha_4\alpha_3)$,\\
   $I^{(8)}=(2+x_1^2+2x_2x_1+x_3x_1+2x_1x_4-2x_1x_5+2x_1-x_2^2+2x_3x_2+2x_4x_2+2x_2x_5-x_2-2x_3x_4+x_3x_5+x_3-x_4^2-x_4x_5-x_4+x_5^2-2x_5, 1-x_1^2+2x_2x_1-x_3x_1+2x_1x_4+2x_1x_5-2x_1+2x_2^2-x_3x_2-x_4x_2-x_2x_5+2x_2+2x_3^2-2x_3x_4+2x_3x_5-2x_4^2+2x_4x_5-x_4-x_5^2+x_5, 1+2x_2x_1+x_3x_1+x_1x_4+2x_1x_5+x_2^2+x_4x_2-2x_2x_5-x_3^2+x_3x_4-x_3x_5-x_3+x_4x_5-x_4+x_5^2-x_5, x_1^2+2x_2x_1+2x_3x_1+2x_1x_4+x_1x_5-2x_1+2x_2^2+x_3x_2+2x_4x_2-2x_2-2x_3^2+2x_3x_4-2x_3-2x_4^2-x_4x_5+x_5^2, x_1-2x_2-2x_3+2x_4+x_5-2)\subset
   \Q[x_1,x_2,x_3,x_4,x_5]$.

   \item[9.]
   $f^{(9)}=(y^2+(\alpha_1-\alpha_2\alpha_4)y+\alpha_2\alpha_5+\alpha_3+\alpha_5)(y^2+y(1+\alpha_3-\alpha_2+\alpha_3\alpha_5)+\alpha_4+\alpha_3-\alpha_1\alpha_5)(y+\alpha_3\alpha_5+\alpha_2\alpha_4\alpha_3)$,\\
   $I^{(9)}=I^{(8)}\subset \Q[x_1,x_2,x_3,x_4,x_5]$.

   \item[10.]
   $f^{(10)}=(y+\alpha_1+\alpha_2+\alpha_6+\alpha_3\alpha_4+\alpha_2\alpha_5\alpha_6)(y+\alpha_2\alpha_6-\alpha_1\alpha_5+2-\alpha_3+\alpha_4)$,\\
   $I^{(10)}=(-1+2x_1^2+x_1x_2+2x_1x_3-x_1x_4+x_1x_5-2x_1x_6+2x_1-x_2^2+x_3x_2+2x_4x_2-2x_2+x_3^2-2x_3x_5+2x_3x_6-2x_4^2-x_4x_5+2x_4x_6+x_4-2x_5x_6-2x_5+x_6^2+2x_6, -2x_1^2-x_1x_5-2x_1-2x_2^2+x_3x_2-x_4x_2-x_2x_6-x_3^2+2x_4x_3+2x_3x_5-x_3+2x_4^2-2x_4-x_5^2-x_5x_6+x_5+2x_6^2, -1+x_4-x_5+x_6+x_5x_6-2x_4x_6+2x_4x_5-2x_3x_6+x_3x_5+x_4x_3-x_2x_6+2x_2x_5+x_3x_2-x_1x_6-x_1x_5-2x_1x_4-x_1x_3-x_1x_2+2x_2-2x_1^2+x_2^2+2x_3^2-x_4^2+x_5^2+x_6^2, 2-x_1^2-2x_1x_3-2x_1x_4-2x_1x_5+x_1x_6+x_1-x_3x_2+2x_4x_2-2x_2x_6-x_2+2x_3^2-x_4x_3+2x_3x_5+x_3x_6+x_3+x_4^2+x_4x_5-x_4-2x_5^2+2x_5x_6+x_5-x_6, x_1x_2-2x_1x_4-x_1x_5-x_1x_6-x_1-x_3x_2-2x_4x_2-2x_2x_5-2x_2x_6-x_3^2-x_4x_3-2x_3x_5-2x_3x_6-x_3-2x_4x_5-2x_4x_6-2x_4-x_5^2+2x_5x_6-x_5+2x_6^2+x_6, -2x_1+x_2+x_3-2x_4+2x_5+x_6-2)\subset
   \Q[x_1,x_2,x_3,x_4,x_5,x_6]$.
\end{itemize}

\end{document}